\documentclass[submission,copyright,creativecommons]{eptcs}
\providecommand{\event}{AiML 2026} 

\usepackage{aiml26}

\usepackage{iftex}

\ifpdf
  \usepackage{underscore}         
  \usepackage[T1]{fontenc}        
\else
  \usepackage{breakurl}           
\fi

\usepackage[utf8]{inputenc}
\usepackage{csquotes}
\usepackage{setspace}
\usepackage[dvipsnames]{xcolor}

\usepackage{hyperref}
\usepackage{amsmath, amssymb, amscd, amsthm, amsfonts}
\usepackage{graphicx}
\usepackage{adjustbox}
\usepackage{hyperref}
\usepackage{cleveref}
\usepackage{amsthm}
\usepackage[utf8]{inputenc}
\usepackage[english]{babel}
\usepackage{proof}

\usepackage{proof}
\usepackage{tikz-cd}
\usepackage{xcolor}
\usepackage{thmtools, thm-restate}
\usepackage{colorblind}
\tikzcdset{scale cd/.style={every label/.append style={scale=#1},
		cells={nodes={scale=#1}}}}

\title{A Semantics for Belief in Simplicial Complexes}
\author{Adam Bjorndahl
\institute{Carnegie Mellon University\\ Pittsburgh, PA}
\email{abjorn@cmu.edu}
\and
Philip Sink
\institute{Carnegie Mellon University\\ Pittsburgh, PA}
\email{psink@andrew.cmu.edu}
}

\newcommand{\titlerunning}{A Semantics for Belief in Simplicial Complexes}
\newcommand{\authorrunning}{A. Bjorndahl, P. Sink}

\hypersetup{
  bookmarksnumbered,
  pdftitle    = {\titlerunning},
  pdfauthor   = {\authorrunning},
  pdfsubject  = {A Semantics for Belief in Simplicial Complexe},               
  pdfkeywords = {Belief,
  	Knowledge,
  	Doxastic,
  	Epistemic,
  	Simplicial Complex,
  	Geometry,
  	Semantics,} 
}

\newcommand{\M}{\mathcal{M}}

\newcommand{\draft}[1]{{\color{red}{\textsc{[#1]}}}}
\newcommand{\commentout}[1]{}
\newcommand{\defin}[1]{\textbf{#1}}

\renewcommand{\phi}{\varphi}

\newcommand{\lthen}{\rightarrow}

\newcommand{\F}{\mathcal{F}}
\renewcommand{\M}{\mathcal{M}}
\renewcommand{\L}{\mathcal{L}}
\newcommand{\N}{\mathcal{N}}

\newcommand{\full}{\textsf{FULL}}

\usepackage{amssymb,tikz}

\newcommand{\mysetminusD}{\hbox{\tikz{\draw[line width=0.6pt,line cap=round] (3pt,0) -- (0,6pt);}}}
\newcommand{\mysetminusT}{\mysetminusD}
\newcommand{\mysetminusS}{\hbox{\tikz{\draw[line width=0.45pt,line cap=round] (2pt,0) -- (0,4pt);}}}
\newcommand{\mysetminusSS}{\hbox{\tikz{\draw[line width=0.4pt,line cap=round] (1.5pt,0) -- (0,3pt);}}}

\newcommand{\mysetminus}{\mathbin{\mathchoice{\mysetminusD}{\mysetminusT}{\mysetminusS}{\mysetminusSS}}}

\begin{document}
\maketitle

\begin{abstract}
We provide novel semantics for belief using simplicial complexes. In our framework, belief is a KD45 modality that satisfies ``knowledge implies belief'' (``If you know phi, then you believe phi''); in addition, we adopt the (standard) assumption that each facet in our simplicial models contains exactly one vertex for each agent. No existing model of belief in simplicial complexes that we are aware of is able to satisfy all of these conditions without trivializing belief to coincide with knowledge. We establish a truth-preserving correspondence between our simplicial framework and standard relational models for knowledge and belief;
this involves, notably, proving that all relational models can be simulated using proper relational models, a result of independent interest.
Finally, we apply these results to provide a simple axiomatization.
\end{abstract}

\section{Introduction} \label{sec:int}

Interpreting multi-agent epistemic logic using simplicial complexes is a recent and thriving area of research.
\cite{SimpBel,KaSC,FA,HG,SimpSet,SimpDEL,Death,DoA} 
Our goal in this paper is to extend this approach to accommodate not only knowledge but also \emph{belief};
this turns out to involve some subtleties, the resolutions of which are interesting in their own right.


To set the stage, we begin by briefly reviewing the basic notion of a simplicial complex and the intended epistemic interpretation. Informally, a \emph{simplex} can be pictured as a multi-dimensional triangle; that is, a collection of vertices where every vertex is connected to every other vertex by an edge. So with three vertices we obtain a standard triangle, four gives a tetrahedron, and in general, $n$ many vertices gives an $(n-1)$-dimensional triangle. A \emph{simplicial complex} can be thought of as a collection of simplices with some perhaps sharing vertices, edges, or larger faces. The simple example below depicts a $1$-dimensional simplex ($bc$) and a $2$-dimensional simplex ($acd$) glued together at a single vertex ($c$):

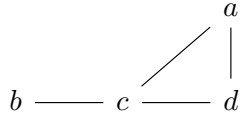
$$
\begin{tikzcd}
	&                                           & {a} \arrow[d, no head] \\
	{b} \arrow[r, no head] & {c} \arrow[r, no head] \arrow[ru, no head] & {d}                   
\end{tikzcd}$$

Formally, a \defin{simplicial complex} is a set $S$ closed under subsets; that is, if $X \subseteq Y \in S$, then $X \in S$.
Intuitively, $S$ is the set of all triangles in the simplicial complex. Closure under subsets corresponds to the idea that any triangle (of any dimension) contains all of the lower-dimensional triangles that make it up; for example, a $2$-dimensional triangle contains its edges (which are $1$-dimensional triangles), and these edges contain their endpoint vertices (which are $0$-dimensional triangles). Accordingly,
every element $X \in S$ is called a \emph{face},
the elements in such an $X$ are called \emph{vertices},
and those faces $X \in S$ that are maximal with respect to set inclusion are called \emph{facets}. For example, the diagram above corresponds to the simplicial complex
$$S = \{\{b,c\}, \{a,c,d\}, \{a,c\}, \{a,d\}, \{c,d\}, \{b\}, \{a\}, \{c\}, \{d\}, \emptyset\},$$
which has facets $\{b,c\}$ and $\{a,c,d\}$.

The basic idea behind interpreting epistemic logic using simplicial complexes is to associate vertices with agents and to view facets as analogous to possible worlds
\cite{KaSC,SimpDEL}.
Intuitively, a vertex associated to agent $a$ represents a particular ``perspective'' of $a$, with the simplicial structure encoding which agent perspectives are epistemically compatible with which others.
Thus, whereas in standard relational semantics $a$'s knowledge is given by quantifying over those worlds that $a$ ``considers possible'' (via an accessibility relation), in the simplicial setting this quantification is instead over those facets that share a vertex associated to agent $a$. We formalize all the details in Section \ref{sec:sem}.

As noted, the goal of this paper is to generalize this semantics to provide a model for multi-agent belief instead of (or in addition to) knowledge.
A key obstacle to this endeavor is that simplicial semantics hardcodes \emph{factivity}
($K\phi \lthen \phi$)
in a way that relational semantics does not.
Naturally, models for belief ought not validate the corresponding principle ($B\phi \lthen \phi$), since agents can have false beliefs.

In relational semantics, knowledge is typically weakened to belief by passing to a subset of those worlds considered possible. More precisely, if $R(w)$ denotes the set of all worlds epistemically accessible from $w$, then a reasonable candidate for a relation $Q$ representing \emph{doxastic} accessibility (i.e., beliefs) is one where $Q \subseteq R$ (so also $Q(w) \subseteq R(w)$). This ensures that knowledge implies belief ($K\phi \lthen B\phi$)
and allows agents to potentially believe falsehoods even though they cannot know falsehoods (because $Q$ need not be reflexive even though $R$ is assumed to be).

Somewhat surprisingly, the obvious analogous approach does not work out in the simplicial setting. Suppose our model begins with a simplicial complex $S$ in the background
(representing knowledge, as usual),
and adds a subcomplex $S' \subseteq S$
to represent the beliefs of
each agent---so to determine an agent's beliefs we would only quantify over facets in $S'$.
Then, as desired, for a given facet $X$ in $S$, it might be the case that $\varphi$ is true at $X$ yet at all facets $X'$ in $S'$ which are $a$-accessible from $X$, $\varphi$ is false.
This produces a counterexample to factivity for belief.
However, at all facets $X'$ in $S'$, every agent \textit{will} have factive beliefs, and as a consequence
this model does not allow agents to consider it possible that any \textit{other} agents have false beliefs;
this is formalized in Proposition \ref{pro:sref}.

Our solution to this issue is to enrich the model with multiple ``belief'' subcomplexes---one for each agent.
This solves the problem outlined above, and leads naturally to the question of how these enriched simplicial models are related to standard relational models for knowledge and belief. This question, too, is not entirely straightforward, since the standard translation between simplicial and relational frameworks requires a ``properness'' condition that turns out to be quite restrictive when it comes to representing beliefs. We discuss this translation, and its restrictiveness, in Section \ref{sec:con}, and then
establish in Section \ref{sec:prop} a general result that allows one to simulate any model with a proper model, effectively dissolving the restriction. This result is of independent interest and as such we attempt to present the details of that section in a relatively self-contained way.

We are not the first to propose a representation of belief using simplicial complexes.
Both ``Knowledge and Simplicial Complexes'' \cite{KaSC} as well as ``Simplicial Belief'' \cite{SimpBel} tackle the same idea.
While each of these approaches is interesting in its own right, they differ substantially from our models both in formal implementation and in the overall logic of belief that they validate. We outline these frameworks and discuss their relationship to our models in Section \ref{sec:Conc}.

The rest of the paper is organized as follows. In Section \ref{sec:sem}, we provide the core definitions, introduce our new model, and prove that in this framework multiple complexes are needed to allow the right kind of inter-agent uncertainty. In Section \ref{sec:con} we establish a translation between
standard relational models for knowledge and belief and
our simplicial representation, highlighting the role of the properness condition, which is thoroughly analyzed in Section \ref{sec:prop}. Section \ref{sec:sac} applies the translation results to prove soundness and completeness. Section \ref{sec:Conc} concludes with a discussion of related and future work.

\section{Semantics for Belief} \label{sec:sem}

The literature on simplicial semantics divides roughly into two approaches to defining valuations for propositional atoms. Loosely speaking, the first approach assigns
truth values directly to the facets
\cite{Death,FA,SimpSet},
while the second ``vertex-based'' approach assigns them (in a partial way) to the vertices and then ``lifts'' them to facets.
\cite{DoA,KaSC,SimpDEL,HG}
While we feel that the latter approach is more interesting from a logical and epistemic perspective, the former is much more technically straightforward and will allow us to better illustrate the specific novelty of the belief models we wish to present in this work, so that is the approach we adopt here. 


Let $Ag$ be a finite set of agents.
A \defin{simplicial frame} is a tuple $(N,V,S)$ where $N$ is a set of \emph{nodes}, $V: N \to Ag$ assigns an agent to each node (this is sometimes called the \emph{coloring function}), and $S \subseteq 2^{N}$ is a simplicial complex. Let $\F(S)$ denote the set of facets of $S$. As in much of the previous literature
\cite{KaSC,SimpDEL},
we henceforth restrict attention to simplicial frames that satisfy the \defin{uniquely colored facets (UCF)} condition: for each $X \in \F(S)$ and each $a \in Ag$, there is exactly one node $n \in X$ such that $V(n) = a$. Note that this implies that all facets of $S$ contain exactly $|Ag|$ nodes, one for each agent.

Let $\mathfrak{P}$ be a countable set of propositional atoms. A \defin{simplicial model} $\M$ is a simplicial frame $(N,V,S)$ together with a \emph{valuation function} $L: \mathfrak{P} \to 2^{\F(S)}$; intuitively, $L(P)$ tells us which facets $P$ is true at. In this setting, the language $\L_{K}(Ag)$ recursively defined by
$$\varphi ::= P \, | \, \bot \, | \, \phi \lthen \psi \, | \, K_{a}\phi,$$
where $P \in \mathfrak{P}$ and $a \in Ag$, can be interpreted in simplicial models as follows:
\begin{align*}
	\M,X & \models P \text{ iff } X \in L(P)\\
	\M,X & \nvDash \bot\\
	\M,X & \models \phi \lthen \psi \text{ iff } \M,X \models \phi \text{ implies } \M,X \models \psi\\
	\M,X & \models K_a \phi \text{ iff } \forall Y \in \F(S) \text{ if } \pi_a(Y) = \pi_a(X) \text{ then } \M,Y \models \phi,
\end{align*}
where $\pi_a(Z) = V^{-1}(a) \cap Z$, that is, $\pi_{a}$ maps each facet to the unique $a$-colored vertex it contains.
The other Boolean connectives are defined in the usual way.

This is the standard simplicial semantics for multi-agent knowledge: agent $a \in Ag$ knows $\phi$ at facet $X$ just in case $\phi$ is true at all facets that share an $a$-colored node with $X$. Note that the factivity of knowledge follows immediately: if $X \models K_{a} \phi$, then since trivially $\pi_{a}(X) = \pi_{a}(X)$, the semantic clause for $K_{a}$ forces us to conclude that $X \models \phi$.

To incorporate belief---which of course we do not want to be factive---we add in additional simplicial complexes. A \defin{simplicial belief model} is a tuple $(N,V,S,\{S_{a}\}_{a \in Ag},L)$, where $(N,V,S,L)$ is a simplicial model and each $S_{a}$ is a nonempty simplicial complex contained in $S$; we further assume that each complex $S_{a}$ satisfies the
UCF condition. Note that this guarantees that $\F(S_{a}) \subseteq \F(S)$. Expand the language recursively to include unary belief modalities $B_{a}$ for each agent $a \in Ag$ and call the result $\L_{KB}(Ag)$.
Moreover, to guarantee
consistency of beliefs (i.e., that $\lnot B_{a} \bot$ is valid),
we also assume that for every $a$-colored node $n$, there is $X\in\F(S_a)$ such that $n\in X$.
Given a simplicial belief model $\M$, we can then enrich the semantics with the following clause:
\begin{align*}
	\M,X & \models B_a \phi \text{ iff } \forall Y \in \F(S_{a}) \text{ if } \pi_a(Y) = \pi_a(X) \text{ then } \M,Y \models \phi.
\end{align*}

So this is just like the clause for knowledge, except for each agent $a$ we quantify only over facets in $S_{a}$ rather than all the facets of $S$. Because of this, belief is not in general forced to be factive, since at any $X \in \F(S) \mysetminus \F(S_{a})$ we can have $X \models B_{a}P$ but $X \nvDash P$ (for example, if $L(P) = \F(S_{a})$).

As discussed in Section \ref{sec:int}, a natural question that arises is why we need a different simplicial complex $S_{a}$ for each agent $a$. After all, simplicial models for multi-agent knowledge make do with just one. The issue is made clear in the following proposition.

\begin{proposition} \label{pro:sref}
	Let $\M$ be a simplicial belief model with $S_{a} = S_{b}$. Then $\M \models B_{a}(B_{b} \phi \lthen \phi)$.
\end{proposition}
\begin{proof}
	Let $X \in \F(S)$, and consider any $Y \in \F(S_{a})$ such that $\pi_{a}(Y) = \pi_{a}(X)$. If $Y \models B_{b} \phi$, then by definition for all $Z \in \F(S_{b})$ such that $\pi_{b}(Z) = \pi_{b}(Y)$ we must have $Z \models \phi$. Since $\F(S_{b}) = \F(S_{a})$ and $Y \in \F(S_{a})$ by assumption, we know that $Y \in \F(S_{b})$, so in fact we must have $Y \models \phi$. This shows that $Y \models B_{b}\phi \lthen \phi$ and thus $X \models B_{a}(B_{b}\phi \lthen \phi)$, as desired.
\end{proof}

In other words, when $S_{a}$ and $S_{b}$ coincide, agent $a$ becomes incapable of considering agent $b$ fallible with respect to their beliefs (and, of course, vice-versa). This is not an issue for multi-agent knowledge because knowledge is not fallible at all (so, naturally, agents know this).

\section{Connection to Relational Semantics} \label{sec:con}

It is well-known \cite{KaSC} that a broad class of relational models for multi-agent knowledge can be transformed into simplicial models in a truth-preserving manner. We briefly review this transformation here in order to illustrate how it works---and where it fails---when belief is added into the mix.

Recall that a \defin{relational model (over $Ag$)} is a tuple $\N = (W,(R_{a})_{a \in Ag},v)$ where $W$ is a nonempty set of \emph{possible worlds}, each $R_{a}$ is a binary relation on $W$ called an \emph{accessibility relation}, and $v: \mathfrak{P} \to 2^{W}$ is a \emph{valuation}. Semantics are given in the usual way:
\begin{align*}
	\N,w & \models P \text{ iff } w \in v(P)\\
	\N,w & \nvDash \bot\\
	\N,w & \models \phi \lthen \psi \text{ iff } \N,w \models \phi \text{ implies } \N,w \models \psi\\
	\N,w & \models K_a \phi \text{ iff } \forall w' \in R_{a}(w) (\N,w' \models \phi),
\end{align*}
where $R_{a}(w) = \{w' \in W \: : \: wR_{a}w'\}$.
A \defin{relational frame} is a model without the valuation.

Except where otherwise noted, we restrict our attention to relational models where each $R_{a}$ is an equivalence relation---such models validate factivity ($K_{a} \phi \lthen \phi$), positive introspection ($K_{a} \phi \lthen K_{a} K_{a} \phi$), and negative introspection ($\lnot K_{a} \phi \lthen K_{a} \lnot K_{a} \phi$)---and refer to these as \defin{relational models for introspective knowledge}. Since these principles are also validated by all simplicial models, we must impose them on relational structures in order to have any hope of a truth-preserving transformation. We write $[w]_{a}$ to denote the equivalence class of $w$ with respect to $R_{a}$.

Given a relational model $\N$ as above, we define a simplicial frame $(N_{\N},V_{\N},S_{\N})$ as follows:
\begin{itemize}
	\item
	$N_{\N} = \{([w]_{a},a) \: : \: w \in W, a \in Ag\}$
	\item
	$V_{\N}([w]_{a},a) = a$
	\item
	$S_{\N} = \{F \subseteq N_{\N} \: : \: \bigcap_{n \in F} e(n) \neq \emptyset\}$,
\end{itemize}
where $e: N_{\N} \to 2^{W}$ is projection to the first coordinate (returning the equivalence class associated with each node). Intuitively, each equivalence class $[w]_{a}$ captures a certain epistemic ``perspective'' of agent $a$, and thus we take our nodes to be all such equivalence classes (labelled by the agent in question). The simplicial structure is then determined by grouping precisely those nodes that correspond to equivalence classes with nonempty intersection, since these are the perspectives that can all obtain simultaneously (i.e., at the worlds in the intersection). Thus, each face of $S_{\N}$ corresponds to nonempty intersection of equivalence classes from distinct agents.

\begin{proposition}
	The simplicial frame $(N_{\N},V_{\N},S_{\N})$ satisfies UCF.
\end{proposition}

\begin{proof}
	We must show that each facet contains exactly one node for each agent. First observe that no face in $S_{\N}$ can contain more than one $a$-colored node, since if $([w]_{a},a)$ and $([w']_{a}, a)$ are distinct, then $[w]_{a}$ and $[w']_{a}$ must also be distinct and thus have empty intersection.
	
	Now suppose that for some $a \in Ag$, $F \in S_{\N}$ does not contain an $a$-colored node. By definition, there exists some $w \in \bigcap_{n \in F} e(n)$; it follows immediately that $[w]_{a} \cap \bigcap_{n \in F} e(n) \neq \emptyset$, and thus $F \cup \{([w]_{a},a)\} \in S_{\N}$. This shows that $F$ is not a facet (since it's not maximal), which in turn implies that any facet must contain an $a$-colored node for every $a \in Ag$. 
\end{proof}

The latter half of the proof above suggests a natural way to associate worlds in $\N$ with facets in the simplicial frame we have constructed: namely, define $f \colon W \to \F(S_{\N})$ by
$$f(w) = \{([w]_{a},a) \: : \: a \in Ag\}.$$
It is not hard to see that $f$ is surjective, but unfortunately, it may not be injective: $f(w) = f(w')$ precisely when, for all $a \in Ag$, $[w]_{a} = [w']_{a}$. In this case there is no clear way to define a valuation on $\F(S_{\N})$ that ``simulates'' the valuation in $\N$, since two worlds that disagree about the truth value of $P$, say, may be associated to the same facet.

This motivates the following definition that has become standard in the literature
\cite{KaSC}:
a relational model is called \defin{proper} if, for all $w \in W$, $\big|\bigcap_{a \in Ag} [w]_{a} \big| = 1$.%
\footnote{This definition of properness differs from standard presentations when the relations are not equivalence relations, as has also been noted in \cite{Lehnherr_2025}.}
From this it easily follows that $f$ is injective. Define $L_{\N}(P) = f[v(P)] = \{f(w) \: : \: w \in v(P)\}$ and let $M_{\N}$ be the simplicial model $(N_{\N},V_{\N},S_{\N}, L_{\N})$. We then have:


\begin{theorem} \label{SimpTrans}
	Let $\N$ be a proper relational model. For all formulas $\phi \in \L_{K}(Ag)$, for all $w \in W$,
	$$\N,w \models \phi \textrm{ iff } \M_{\N}, f(w) \models \phi.$$
\end{theorem}


\begin{proof}

Induction on formulas. Let $P\in\mathfrak{P}$. First suppose $\N,w\models P$;
then
$w\in v(P)$, so
$f(w) \in f[v(P)] = L_{\N}(P)$,
which yields $\M_{\N},f(w)\models P$, as desired. Conversely, suppose $\M_{\N},f(w)\models P$; then
$f(w)\in L_{\N}(P) = f[v(P)]$.
Because $f$ is injective, we must then have $w\in v(P)$, which yields $\N,w\models P$ as desired.

The case for $\bot$ is trivial.

Now suppose inductively
that
for all $w \in W$, we have that
$\N,w \models \phi \textrm{ iff } \M_{\N}, f(w) \models \phi$
and
$\N,w \models \psi \textrm{ iff } \M_{\N}, f(w) \models \psi$.
We wish to show that this equivalence holds also for $\phi \lthen \psi$ and for $K_{a} \phi$. For the former, observe that
$\N,w\models\varphi\rightarrow\psi$
if and only if $\N,w \models \varphi$ implies $\N,w\models\psi$,
which (by the inductive hypothesis)
is true if and only if $\M_{\N},f(w)\models\varphi$ implies $\M_{\N},f(w)\models\psi$. And of course this is true if and only if $\M_{\N},f(w)\models\varphi\rightarrow\psi$, as desired.

At last we turn to the knowledge modalities.
Suppose $\N,w\models K_a\varphi$. Fix arbitrary $Y\in\mathcal{F}(S_\N)$ such that $\pi_a(Y)=\pi_a(f(w))$. By definition of $f$, $\pi_a(Y)=\pi_a(\{([w]_{a},a) \: : \: a \in Ag\})=([w]_{a},a)$.
Because $f$ is surjective, we can find
$w'\in W$ such that $f(w')=Y$.
So $([w']_{a},a) = \pi_a(f(w'))=\pi_a(Y)=([w]_{a},a)$, and therefore $[w']_a=[w]_a$, meaning $w'\in[w]_a$. We know from our initial supposition that for all $w'\in[w]_a$, $\N,w'\models\varphi$. So, by the inductive hypothesis, we have $\M_{\N},f(w')\models\varphi$, hence $\M_{\N},Y\models\varphi$; since $Y$ was chosen arbitrarily, this yields $\M_{\N},f(w)\models K_a\varphi$.

Conversely, suppose $\M_{\N},f(w)\models K_a\varphi$.
By definition,
for all $Y\in\mathcal{F}(S_\N)$ such that $\pi_a(Y)=\pi_a(f(w))$, $\M_{\N},Y\models\varphi$. 
Let 
$w'\in[w]_a$.
Then, since $[w']_{a} = [w]_{a}$, we must have
$\pi_a(f(w))=\pi_a(f(w'))$, so $\M_{\N},f(w')\models\varphi$. By the inductive hypothesis, then, $\N,w'\models\varphi$, and since $w'$ was chosen arbitrarily, we conclude $\N,w\models K_a\varphi$.

\end{proof}

The existing literature has largely taken properness to be a technical condition, introduced in order to facilitate translations from relational to simplicial models.
\cite{KaSC,DoA,Death,SimpDEL}
Whether it is a \textit{reasonable} restriction, in the sense of being epistemically motivated or defensible, is debatable (and actively debated by the authors of this paper!).\footnote{One could argue that it eliminates a kind of ``redundancy of perspectives'', though this depends on a further philosophical supposition that every relevant fact about a world is captured in some agent's perspective.} What is less contentious is that once we incorporate belief into the models, properness is far too strong a requirement. A concrete example will help to illustrate this.




A \defin{relational model for introspective knowledge and belief (over $Ag$)} is just a relational model for introspective knowledge supplemented with new relations $Q_{a} \subseteq R_{a}$, one for each $a \in Ag$; we require that each $Q_{a}$ be serial and constant on $R_{a}$-equivalence classes: that is, if $w R_{a} w'$ then $Q_{a}(w) = Q_{a}(w')$. These new relations are used to interpret belief modalities in the language $\L_{KB}(Ag)$ in the usual way, namely:
\begin{align*}
	w & \models B_a \phi \text{ iff } \forall w' \in Q_{a}(w) (w' \models \phi).
\end{align*}
Under these semantics, the properties of $Q_{a}$ listed above ensure that our models validate knowledge implies belief ($K_{a} \phi \lthen B_{a} \phi$), consistency of belief ($B_{a} \phi \lthen \lnot B_{a} \lnot \phi$), strong positive introspection for belief ($B_{a} \phi \lthen K_{a} B_{a} \phi$), and strong negative introspection for belief ($\lnot B_{a} \phi \lthen K_{a} \lnot B_{a} \phi)$.

Consider the model depicted in Figure \ref{fgr:3ag} for three agents $Ag = \{a,b,c\}$; the red, blue, and green arrows display the relations $Q_{a}$, $Q_{b}$, and $Q_{c}$, respectively. In particular, at all worlds $w$, agent $a$ considers only $w_{1}$ and $w_{2}$ possible, agent $b$ is sure that the true world is $w_{1}$, and agent $c$ is sure that the true world is $w_{2}$. Note that we have not explicitly depicted the relations $R_{a}$, $R_{b}$, or $R_{c}$, but because each of these must be an equivalence relation containing the corresponding belief relation, there is only one possibility: they are all the complete relation. Thus, the corresponding relational model for introspective knowledge is not proper and so cannot be translated into a simplicial model.

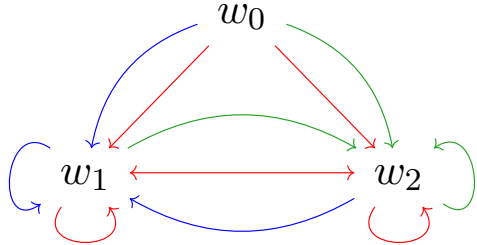
\begin{figure}[htbp]
	\begin{center}
		\begin{tikzcd}[scale cd=1.5]
			& w_0 \arrow[red, ldd] \arrow[red, rdd] \arrow[blue, ldd, bend right] \arrow[green, rdd, bend left] &                                                                                                                  \\
			&                                                                            &                                                                                                                  \\
			w_1 \arrow[red, rr, leftrightarrow] \arrow[green, rr, bend left] \arrow[red, loop, distance=2em, in=305, out=235] \arrow[blue, loop, distance=2em, in=215, out=145] &                                                                            & w_2 \arrow[blue, ll, bend left] \arrow[red, loop, distance=2em, in=305, out=235] \arrow[green, loop, distance=2em, in=35, out=325]
		\end{tikzcd}
		\caption{A 3-agent relational frame for introspective knowledge and belief} \label{fgr:3ag}
	\end{center}
\end{figure}\label{Figure1}

However, there does not seem to be anything trivial or redundant about this model of belief; as such, an inability to translate it into the simplicial setting is a substantive restriction. Fortunately, while this model itself cannot be so translated, there does exist an equivalent (in the sense of bisimilarity) model that can be. This follows from
the results of the next section.

\section{Properness} \label{sec:prop}

The results of this section are relevant outside the context of the simplicial semantics for belief that is the central project of this paper; as such, we present them here first for simplicial models generally and only subsequently extend them to simplicial belief models.

Let $\N = (W,(R_{a})_{a \in Ag},v)$ be a relational model with $Ag = \{1, \ldots, n\}$; in this section we do \textit{not} assume that the relations $R_{i}$ are equivalence relations (or have any properties). As defined above, $\N$ is called \defin{proper} if it contains no pair of distinct points $w,w' \in W$ such that for all $i \in Ag$, $wR_{i}w'$. As noted, this properness assumption has been deployed repeatedly in the literature to facilitate translation from relational models to simplicial models \cite{KaSC,DoA,Death,SimpDEL}, but to the best of our knowledge there has been no systematic study of just how restrictive it is. In this section we show that in fact properness is not restrictive at all: \textit{every} relational model is equivalent (via bisimulation) to a proper relational model.%
\footnote{The only other formal discussion of this topic that we are aware of occurs in \cite{Death}, where the translation of a particular ``canonical'' Kripke model, which is not itself proper, is conjectured to be bisimilar to a proper Kripke model using a broad ``unwinding'' method. Unfortunately, this technique would not, in general, preserve properties like transitivity and symmetry, limiting its usefulness in the context of simplicial semantics, where these properties are typically necessary.}

We begin with the case where $W$ is finite; say $W = \{w_{1}, \ldots, w_{m}\}$. We will construct a new relational model $\tilde{\N} = (\tilde{W}, (\tilde{R}_{i})_{i=1}^{n}, \tilde{v})$, prove that it is proper, and exhibit a surjective, bounded morphism from $\tilde{\N}$ to $\N$.\footnote{For a review of bounded morphisms, see for example \cite{Bjorndahl2024}.}

Set $\tilde{W} = W \times W$; it is helpful to picture this new state space as consisting of $m$ disjoint copies of the original set $W$, namely,
$$\tilde{W} = W \times \{w_{1}\} \cup \cdots \cup W \times \{w_{m}\},$$
so the point $(w_{j}, w_{k}) \in \tilde{W}$ may be thought of as the $j$th element of the $k$th copy of $W$. Accordingly, we define
$$\tilde{v}(P) = \{(w_{j},w_{k}) \in \tilde{W} \: : \: w_{j} \in v(P)\};$$
in other words, the primitive propositions true at $(w_{j},w_{k})$ in $\tilde{\N}$ are precisely those that are true at $w_{j}$ in $\N$.

Lastly we must define the relations, and we do so with the idea of preserving this correspondence in truth between $(\tilde{\N}, (w_{j}, w_{k}))$ and $(\N, w_{j})$ (indeed, we will ultimately show that projection to the first coordinate is the promised bounded morphism). For $i > 1$, define
$$(w_{j},w_{k}) \tilde{R}_{i} (w_{j'},w_{k'}) \textrm{ iff $k = k'$ and $w_{j} R_{i} w_{j'}$}.$$
Notice that if we imposed this same definition for $i=1$, then $\tilde{\N}$ would simply be the disjoint union of $m$ copies of $\N$. Instead, we will define $\tilde{R}_{1}$ using a different partition of $\tilde{W}$:
let $=_m$ denote equality modulo $m$, and for each $\ell \in \{0, \ldots, m-1\}$, let
$$\tilde{W}_{\ell} = \{(w_{j}, w_{k}) \in \tilde{W} \: : \: k - j =_m \ell\},$$
Thus, $\tilde{W}_{0} = \{(w_{1}, w_{1}), (w_{2}, w_{2}), \ldots, (w_{m}, w_{m})\}$ is the ``diagonal'', $\tilde{W}_{1} = \{(w_{1}, w_{2}), (w_{2}, w_{3}), \ldots, (w_{m}, w_{1})\}$, and so on. It is then easy to check the following.

\begin{proposition}
The collection $\{\tilde{W}_{\ell} \: : \: 0 \leq \ell \leq m-1\}$ partitions $\tilde{W}$. Moreover, for each $\tilde{W}_{\ell}$, projection to the first component is a bijection between $\tilde{W}_{\ell}$ and $W$.
\end{proposition}
\begin{proof}
	Suppose $(w_{j},w_{k}) \in \tilde{W}_{\ell} \cap \tilde{W}_{\ell'}$
Then $\ell=_mk-j=_m\ell'$. Since $\ell,\ell'\in\{0,\ldots,m-1\}$, we have that $\ell=\ell'$.
This shows that the $\tilde{W}_{\ell}$ are mutually disjoint.
Moreover, given any $(w_{j},w_{k}) \in \tilde{W}$,
clearly $(w_{j},w_{k})\in\tilde{W}_{\ell}$ for $\ell =_m k-j$; thus,
$\{\tilde{W}_{\ell} \: : \: 0 \leq \ell \leq m-1\}$ partitions $\tilde{W}$.
	
	Define $\rho \colon \tilde{W} \to W$ by $\rho(w_{j},w_{k}) = w_{j}$, namely, projection to the first component.
	Let $(w_{j},w_{k}), (w_{j'},w_{k'}) \in \tilde{W}_{\ell}$, and suppose $\rho(w_{j},w_{k}) = \rho(w_{j'},w_{k'})$, so $w_{j}=w_{j'}$ (which means that $j = j'$). Then, since $k-j =_m \ell =_m k'-j'$, we have that $k=k'$ and so $w_{k}=w_{k'}$. This shows $\rho$ is injective when restricted to any $\tilde{W}_{\ell}$. Next, given $w_j \in W$, let $k \in \{1, \ldots, m\}$ be such that $k =_m j+\ell$; then by definition $(w_{j},w_{k}) \in \tilde{W}_{\ell}$ and $\rho(w_{j},w_{k}) = w_j$, which shows that $\rho$ is surjective when restricted to any $\tilde{W}_{\ell}$.
\end{proof}

Having established this new way of breaking $\tilde{W}$ into copies of $W$, we define
$$(w_{j},w_{k}) \tilde{R}_{1} (w_{j'},w_{k'}) \textrm{ iff $k - j =_m k' - j'$ and $w_{j} R_{1} w_{j'}$}.$$
So, like the other relations, $\tilde{R}_{1}$ is also $m$ disjoint copies of $R_{1}$, but skewed across a different partition: one copy on each set $\tilde{W}_{\ell}$.

\begin{proposition}
$\tilde{\N}$ is proper.
\end{proposition}
\begin{proof}
	Suppose $(w_{j},w_{k}) \tilde{R}_{i} (w_{j'},w_{k'})$ for each $i \in \{1, \ldots, n\}$. Taking $i=1$, we know that $k-j=_mk'-j'$. Taking $i>1$, we know that $k=k'$. So, $j=j'$, and therefore $(w_{j},w_{k})=(w_{j'},w_{k'})$.
\end{proof}

\begin{proposition}
Projection to the first coordinate is a surjective, bounded morphism from $\tilde{\N}$ to $\N$.
\end{proposition}
\begin{proof}
	We have already seen that $\rho$ is surjective, and it
	is clear that $\rho$ also preserves the truth values of atomic formulas. To show $\rho$ is a bounded morphism,
	we need to establish the ``back'' and ``forth'' conditions for the relations $R_{i}$,
	for which the only interesting case is $i=1$.
	
	First, the ``back'' condition:
	let $(w_{j},w_{k}) \in \tilde{W}$ and $w_{j'}\in W$, and suppose $\rho(w_{j},w_{k})R_1w_{j'}$. Fix
	$\ell$
	such that $(w_{j},w_{k}) \in \tilde{W}_{\ell}$.
	Choose $k' \in \{1, \ldots, m\}$ such that $k' =_m j' + \ell$, and consider $(w_{j'},w_{k'})$. We have that $k'-j' =_m \ell$, so $(w_{j'},w_{k'}) \in \tilde{W}_{\ell}$. It follows that $(w_{j},w_{k})\tilde{R}_{1}(w_{j'},w_{k'})$ and of course $\rho(w_{j'},w_{k'}) = w_{j'}$, as desired.
	
	For the ``forth'' condition, suppose that $(w_{j},w_{k})\tilde{R}_{1}(w_{j'},w_{k'})$; then it is immediate from the definition of $\tilde{R}_{1}$ that $w_{j}R_{1}w_{j'}$.
\end{proof}

The construction for countable $W$ is quite analogous, and perhaps even more straightforward since there is no need to appeal to modular arithmetic. In this case we have a bijection $g \colon W \to \mathbb{Z}$. Define $\tilde{\N} = (\tilde{W}, (\tilde{R}_{i})_{i=1}^{n}, \tilde{v})$ by setting
\begin{itemize}
\item
$\tilde{W} = W \times W$;
\item
$\tilde{v}(P) = \{(x,y) \in \tilde{W} \: : \: x \in v(P)\}$;
\item
for $i>1$,
$$(x,y) \tilde{R}_{i} (x',y') \textrm{ iff $y=y'$ and $x R_{i} x'$},$$
and
$$(x,y) \tilde{R}_{1} (x',y') \textrm{ iff $g(y) - g(x) = g(y') - g(x')$ and $x R_{1} x'$}.$$
\end{itemize}

Once again, we have: 
\begin{proposition}
$\tilde{\N}$ is proper, and projection to the first coordinate is a surjective, bounded morphism from $\tilde{\N}$ to $\N$.
\end{proposition}
\begin{proof}
For each $\ell \in \mathbb{Z}$, 
define $\tilde{W}_{\ell}  = \{(x,y) \in \tilde{W} \: : \: g(y) - g(x) = \ell\}$. Analogous to
the above, it is easy to see that the collection $\{\tilde{W}_{\ell} \: : \: \ell \in \mathbb{Z}\}$ is a partition of $\tilde{W}$, and moreover that $\rho \colon \tilde{W}_{\ell} \to W$ is a bijection for each $\ell$.

The proof of properness is also completely parallel to the finite case:
suppose
$(x,y) \tilde{R}_{i} (x',y')$ for each
$i \in \{1, \ldots, n\}$.
Taking $i=1$, we know that
$g(y) - g(x) = g(y') = g(x')$;
taking $i>1$, we know that
$y = y'$.
It follows that $g(x) = g(x')$, so $x = x'$ and $(x,y) = (x',y')$.

It is clear as before that $\rho$ is surjective and preserves the truth values of atomic formulas. 
Thus, to show that it is a bounded morphism, we just need to establish the ``back'' and ``forth'' conditions for the relations $R_{i}$; once again,
the only interesting case is $i=1$.

As above, ``forth'' is immediate from the definition of $\tilde{R}_{1}$. For ``back'', let $(x,y) \in \tilde{W}$ and suppose that $x R_{1} x'$. Let $y' = g^{-1}(g(y) - g(x) + g(x'))$; then $g(y') - g(x') = g(y) - g(x)$, so by definition we have $(x,y) \tilde{R}_{1} (x',y')$ and of course $\rho(x',y') = x'$, as desired.
\end{proof}

This translation can also be extended to continuum-sized models in the obvious way: namely, by replacing $\mathbb{Z}$ with $\mathbb{R}$ in the proof above (literally every other part of the proof remains the same). This is useful, for example, if one wishes to apply the transformation to a canonical model, which is typically of size continuum.\footnote{For models of other cardinalities, we could again duplicate the proof given above if we had a structure like $\mathbb{Z}$ or $\mathbb{R}$ but with the needed cardinality.
And indeed, such a structure is guaranteed to exist by the L\"owenheim-Skolem theorem.
In fact,
we only really need
to apply the L\"owenheim-Skolem theorem to the group axioms, and not $Th(\mathbb{R})$;
the proof above uses only the existence of additive inverses and the fact that $a-b=a-c$ implies $b=c$.}

Finally, we observe that many properties of the relations $R_{i}$ are preserved by this translation: for example, it is easy to see that if $R_{i}$ is reflexive, symmetric, transitive, serial, or Euclidean, then $\tilde{R}_{i}$ is as well---in all cases the proof follows immediately from the fact that $\tilde{R}_{i}$ can be represented as a disjoint union of ``copies'' of $R_{i}$. Thus we have the following summary result:

\begin{theorem} \label{thm:note}
	Let $\N$ be a relational model. Then there is a relational model $\tilde{\N}$ which is proper and a surjective, bounded morphism from $\tilde{\N}$ to $\N$. Furthermore, if $\N$ is
	reflexive (respectively: symmetric, transitive, serial, Euclidean)
	then $\tilde{\N}$ is
	also reflexive (respectively: symmetric, transitive, serial, Euclidean).
\end{theorem}

With that established, we return to our consideration of belief.
Roughly speaking, by expanding this result to apply also to belief, we can
translate the model in Figure \ref{fgr:3ag} into a
proper model that ``simulates'' it via an appropriate surjective, bounded morphism.
This will yield
the model depicted in Figure \ref{fgr:3agprop}.

\begin{figure}[htbp]
	\begin{center}
		\begin{tikzcd}[scale cd=0.5]
			&                                                                                           &                                                                                                                          &                                                                                                                                                  & {(w_0,w_0)} \arrow[blue, ldd, bend right] \arrow[green, rdd, bend left] \arrow[red, lllldddddd, bend right] \arrow[red, rrrrdddddd, bend left] &                                                                                                                                                  &                                                                                                                           &                                                                                           &                                                                                                                          \\
			&                                                                                           &                                                                                                                          &                                                                                                                                                  &                                                                                                                         &                                                                                                                                                  &                                                                                                                           &                                                                                           &                                                                                                                          \\
			&                                                                                           &                                                                                                                          & {(w_1,w_0)} \arrow[green, rr, bend left] \arrow[red, loop, distance=2em, in=305, out=235] \arrow[blue, loop, distance=2em, in=215, out=145] \arrow[red, ldddd, leftrightarrow] &                                                                                                                         & {(w_2,w_0)} \arrow[blue, ll, bend left] \arrow[green, loop, distance=2em, in=35, out=325] \arrow[red, loop, distance=2em, in=305, out=235] \arrow[red, rdddd, leftrightarrow] &                                                                                                                           &                                                                                           &                                                                                                                          \\
			&                                                                                           &                                                                                                                          &                                                                                                                                                  &                                                                                                                         &                                                                                                                                                  &                                                                                                                           &                                                                                           &                                                                                                                          \\
			& {(w_0,w_1)} \arrow[blue, ldd, bend right] \arrow[green, rdd, bend left] \arrow[red, rrrruu] \arrow[red, rrrrrdd] &                                                                                                                          &                                                                                                                                                  &                                                                                                                         &                                                                                                                                                  &                                                                                                                           & {(w_0,w_2)} \arrow[blue, ldd, bend right] \arrow[green, rdd, bend left] \arrow[red, lllluu] \arrow[red, llllldd] &                                                                                                                          \\
			&                                                                                           &                                                                                                                          &                                                                                                                                                  &                                                                                                                         &                                                                                                                                                  &                                                                                                                           &                                                                                           &                                                                                                                          \\
			{(w_1,w_1)} \arrow[green, rr, bend left] \arrow[red, loop, distance=2em, in=305, out=235] \arrow[blue, loop, distance=2em, in=215, out=145] \arrow[red, rrrrrrrr, leftrightarrow, bend right] &                                                                                           & {(w_2,w_1)} \arrow[blue, ll, bend left] \arrow[red, loop, distance=2em, in=305, out=235] \arrow[green, loop, distance=2em, in=35, out=325] &                                                                                                                                                  &                                                                                                                         &                                                                                                                                                  & {(w_1,w_2)} \arrow[green, rr, bend left] \arrow[red, loop, distance=2em, in=305, out=235] \arrow[blue, loop, distance=2em, in=215, out=145] &                                                                                           & {(w_2,w_2)} \arrow[blue, ll, bend left] \arrow[red, loop, distance=2em, in=305, out=235] \arrow[green, loop, distance=2em, in=35, out=325]
		\end{tikzcd}
		\caption{A translation of the model from Figure \ref{fgr:3ag} into a proper model}
		\label{fgr:3agprop}
	\end{center}
\end{figure}
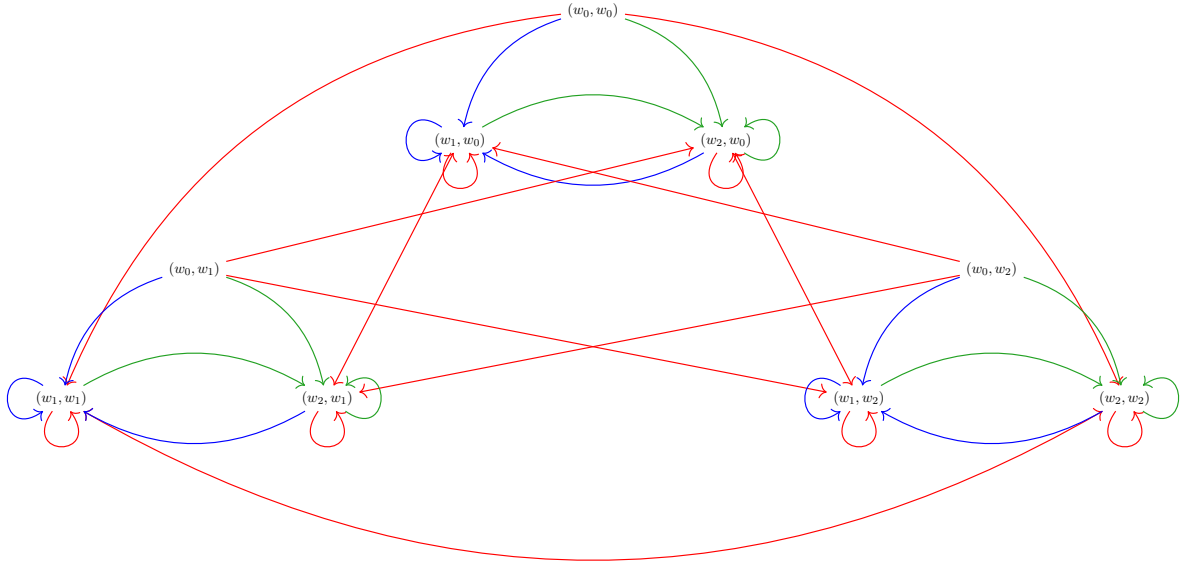\label{fgr:tra}


Because
the relational model in
Figure \ref{fgr:3agprop} is proper, we are able to transform it into a simplicial model
using (a slight generalization of)
Theorem \ref{SimpTrans}. The result of this transformation is shown in Figure \ref{fgr:3agsimp}.
In the next section we
fill in the details---namely, generalizing Theorems \ref{thm:note} and \ref{SimpTrans} to relational models for knowledge \textit{and belief}---%
and use these results to establish a novel completeness result.

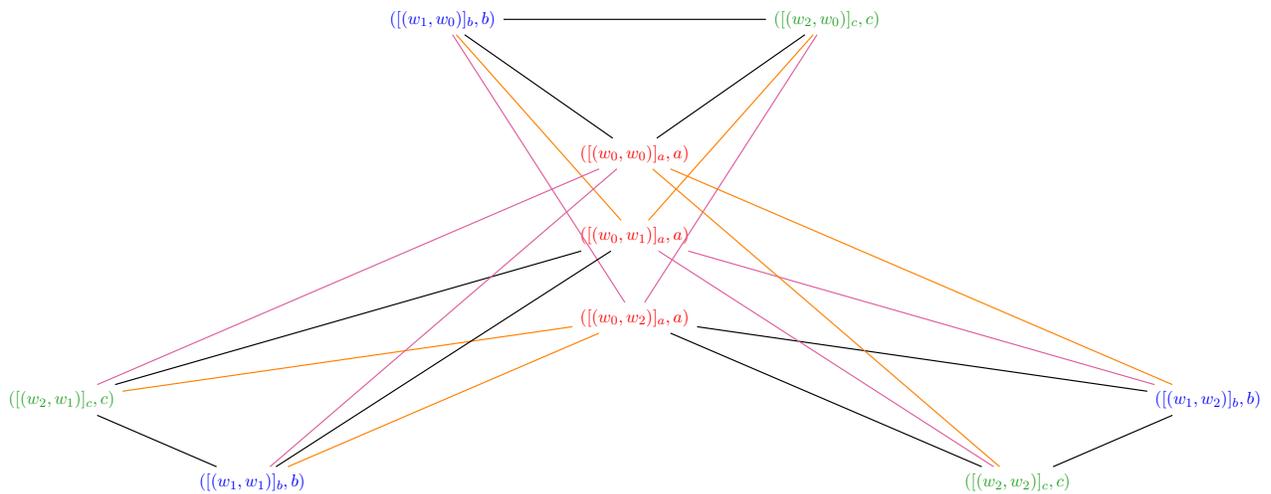
\begin{figure}[htbp]
	\begin{center}
		\begin{tikzcd}[scale cd=0.6]
			&                                                                                                                          & {\color{blue}([(w_1,w_0)]_b,b)} \arrow[rr, no head] \arrow[rdd, no head] \arrow[yellow, rddd, no head] \arrow[violet, rdddd, no head] &                                & {\color{green}([(w_2,w_0)]_c,c)} \arrow[ldd, no head] \arrow[yellow, lddd, no head] \arrow[violet, ldddd, no head] &                                                                                                       &                                                                                                                          \\
			&                                                                                                                          &                                                                                                                       &                                &                                                                                                    &                                                                                                       &                                                                                                                          \\
			&                                                                                                                          &                                                                                                                       & {\color{red}([(w_0,w_0)]_a,a)} &                                                                                                    &                                                                                                       &                                                                                                                          \\
			&                                                                                                                          &                                                                                                                       & {\color{red}([(w_0,w_1)]_a,a)} &                                                                                                    &                                                                                                       &                                                                                                                          \\
			&                                                                                                                          &                                                                                                                       & {\color{red}([(w_0,w_2)]_a,a)} &                                                                                                    &                                                                                                       &                                                                                                                          \\
			{\color{green}([(w_2,w_1)]_c,c)} \arrow[violet, rrruuu, no head] \arrow[rrruu, no head] \arrow[yellow, rrru, no head] &                                                                                                                          &                                                                                                                       &                                &                                                                                                    &                                                                                                       & {\color{blue}([(w_1,w_2)]_b,b)} \arrow[ld, no head] \arrow[yellow, llluuu, no head] \arrow[violet, llluu, no head] \arrow[lllu, no head] \\
			& {\color{blue}([(w_1,w_1)]_b,b)} \arrow[lu, no head] \arrow[violet, rruuuu, no head] \arrow[rruuu, no head] \arrow[yellow, rruu, no head] &                                                                                                                       &                                &                                                                                                    & {\color{green}([(w_2,w_2)]_c,c)} \arrow[yellow, lluuuu, no head] \arrow[violet, lluuu, no head] \arrow[lluu, no head] &                                                                                                                         
		\end{tikzcd}
		\caption{A translation of the model from Figure \ref{fgr:3agprop} into a simplicial model (violet facets belong to $S_a$ and $S_b$, yellow facets belong to $S_a$ and $S_c$, while black facets belong to only $S$)}
		\label{fgr:3agsimp}
	\end{center}
\end{figure}\label{Figure3}

Before turning to the technical details, it
is worth returning to reflect on the properness condition in light of this example.
Intuitively, the model shown in Figure \ref{fgr:3ag} encodes a single perspective for each agent: the set of doxastically accessible worlds is constant for each of $Q_{a}$, $Q_{b}$, and $Q_{c}$, after all.
However, with only one perspective for every agent, simplicial semantics only has the resources to produce a single facet, which is not enough to capture this scenario.
As such, any technique for translating relational models into simplicial models must sometimes proliferate agential perspectives; indeed, Figure \ref{fgr:3agprop} has 9 worlds instead of 3 and intuitively includes 3 ``perspectives'' for each agent, each of which is a ``copy'' of that agent's original, single perspective (up to formula satisfaction). 
In Section \ref{sec:Conc}, we discuss how the more general notion of \emph{simplicial sets} may be a useful framework for avoiding this redundancy of representation.

%
%
%
%
%
%
%
%
%
%

\section{Soundness and Completeness} \label{sec:sac}

Recall that relational models for introspective knowledge and belief validate knowledge implies belief ($K_{a} \phi \lthen B_{a} \phi$),
strong positive introspection
($B_{a} \phi \lthen K_{a} B_{a} \phi$), and strong negative introspection
($\lnot B_{a} \phi \lthen K_{a} \lnot B_{a} \phi)$.
\commentout{
	In conjunction with propositional tautologies, the \textbf{S5} axioms for
	knowledge, the \textbf{KD45} axioms for belief,
	and both belief and knowledge necessitation as well as modus ponens as rules of inference, we call this collection of axioms
	\textbf{FULL}, and the deduction system using them $\vdash_{\textbf{FULL}}$.
}%
Let $\full$ be the axiom system obtained by adding these three axiom schemes to the system \textsf{S5} for knowledge combined with \textsf{KD45} for belief.\footnote{For an overview of these standard axiom systems, see for example \cite{Bjorndahl2024}.}

\begin{theorem}\label{Thm:Comp}
	$\full$ is sound and complete with respect to the class of all simplicial belief models.
\end{theorem}
\begin{proof}
	
	We prove the result in much the same way as has been done for comparable results in previous literature.
	\cite{KaSC,SimpDEL}
	Soundness is straightforward so we focus on completeness.
	We first generate a canonical relational model for introspective knowledge and belief, using maximal
	\full-consistent sets in the usual way.
	Next we transform this canonical model into an equivalent, proper model using a generalization of
	Theorem \ref{thm:note}.
	Finally, we translate this proper model into a simplicial belief model using essentially the same technique as in Section \ref{sec:con}.
	
	Let $W$ be the collection of maximal
	\full-consistent
	sets. For arbitrary $a\in Ag$, we define $R_a$ as follows: $wR_aw'$ if and only if for all formulas $\varphi$, if $K_{a}\varphi\in w$, then $\varphi\in w'$. Similarly, let $Q_a$ be
	given by $wQ_aw'$ if and only if for all formulas $\varphi$, if $B_{a}\varphi\in w$, then $\varphi\in w'$.
	Lastly, define
	$v(P)=\{w \in W \: : \: P \in w\}$.
	Set $\N = (W,(R_a)_{a\in Ag},(Q_a)_{a\in Ag}, v)$.
	
	It is well known that with the above definitions for relations, and the fact that the \textsf{S5} axioms for knowledge are included in $\full$, each $R_a$ is an equivalence relation.
	Indeed, it is straightforward to check that $\N$ is a relational model for introspective knowledge and belief. For example, the fact that $K_a\varphi\rightarrow B_a\varphi \in \full$ forces $Q_{a} \subseteq R_{a}$.
	We will show
	here in detail
	that each $Q_a$ is constant on
	$R_{a}$-equivalence classes, and leave the remaining, easy checks to the reader.
	Suppose $wR_aw'$. Suppose further that $wQ_au$, and $B_{a}\varphi\in w'$.
	From strong positive introspection it follows
	that $K_aB_a\varphi\in w'$. Since $wR_aw'$, we have that $w'R_aw$, and therefore $B_a\varphi\in w$. Since $wQ_au$, it follows that $\varphi\in u$.
	Since $\phi$ was arbitrary, this shows that $w' Q_{a} u$, which
	suffices to show that $Q_a$ is constant on $R_a$-equivalence classes.
	
	It is easy to show, in the standard way, that $\N,w \models \varphi$ if and only if $\varphi \in w$.
	Thus $\N$ refutes every non-theorem of \full.
	Our next step is to find a \textit{proper} relational model for introspective knowledge and belief that does the same; for this we
	extend the results of Section \ref{sec:prop}.
	
	
	
	Since our language is countable, it follows that $W$ is of size continuum. Let $g \colon W \to \mathbb{R}$ be a bijection. Define $\tilde{W} = W^{2}$ and fix a distinguished $b \in Ag$; for all $a \neq b$, define
	\begin{eqnarray*}
		(w,u) \tilde{R}_{a} (w',u') & \textrm{iff} & u = u' \textrm{ and } wR_{a}w'\\
		(w,u) \tilde{Q}_{a} (w',u') & \textrm{iff} & u = u' \textrm{ and } wQ_{a}w'\\
	\end{eqnarray*}%
	and set
	\begin{eqnarray*}
		(w,u) \tilde{R}_{b} (w',u') & \textrm{iff} & g(u) - g(w) = g(u') - g(w) \textrm{ and } wR_{b}w'\\
		(w,u) \tilde{Q}_{b} (w',u') & \textrm{iff} & g(u) - g(w) = g(u') - g(w) \textrm{ and } wQ_{b}w'\\
	\end{eqnarray*}
	Finally, define $\tilde{v}(P) = \{(x,y) \in \tilde{W} \: : \: x \in v(P)\}$, and let $\tilde{\N} = (\tilde{W}, (\tilde{R}_{a})_{a \in Ag}, (\tilde{Q}_{a})_{a \in Ag}, \tilde{v})$.
	Thus, just as in Section \ref{sec:prop},
	we are making $|W|$-many copies of the model $\N$, but for one special agent, $b$, we ``skew'' their accessibility relations (both $R_{b}$ and $Q_{b}$) so they connect worlds across copies rather than within copies. This guarantees that (except for reflexive edges), agent $b$'s accessibility relations never coincide with any other agent's, which in turn ensures the model is proper.
	
	
	Since $Q_{a} \subseteq R_{a}$ for all $a \in Ag$, the definitions above guarantee that also $\tilde{Q}_{a} \subseteq \tilde{R}_{a}$ for all $a \in Ag$. It's also easy to check that for all $a \in Ag$, $\tilde{R}_{a}$ is an equivalence relation and $\tilde{Q}_{a}$ is serial and
	constant on $\tilde{R}_{a}$-equivalence classes.
	Theorem \ref{thm:note}
	now guarantees $\tilde{\N}$ is proper (with respect to the relations $\tilde{R}_{a}$)
	and moreover, that the map $\rho: \tilde{W} \to W$ given by projection to the first coordinate is a surjective, bounded morphism (again, with respect to the $\tilde{R}_{a}$ relations). But it is easy to check that $\rho$ is also a bounded morphism with respect to the $\tilde{Q}_{a}$ relations.
	It follows that every formula refuted in $\N$ is also refuted in $\tilde{\N}$, as desired.

	The final step is to transform $\tilde{\N}$ into an equivalent simplicial
	belief
	model. Because $\tilde{\N}$ is proper and each $\tilde{R}_{a}$ is an equivalence relation, we can define
	$N_{S_{\tilde{\N}}}$, $V_{S_{\tilde{\N}}}$, $S_{S_{\tilde{\N}}}$, and $L_{S_{\tilde{\N}}}$ as we did above,
	in the lead-up to
	Theorem \ref{SimpTrans}. What's missing are the belief subcomplexes, which we define as follows:
	$$\mathcal{F}(S_{\tilde{\N},a}) = \{F\in\mathcal{F}(S_{\tilde{\N}}) \: : \: f^{-1}(F) \tilde{Q}_{a} f^{-1}(F)\},$$
	recalling that $f \colon \tilde{W} \to \F(S_{\tilde{\N}})$ is the bijective correspondence between worlds in $\tilde{N}$ and facets in $S_{S_{\tilde{\N}}}$.
	Specifying the facets of $S_{\tilde{\N},a}$ uniquely determines the simplicial complex itself; it also guarantees that $S_{\tilde{\N},a}$ is a UCF subcomplex of $S_{\tilde{\N}}$. 
Moreover, fix some arbitrary a-colored node $n\in N_{S_{\tilde{\N}}}$. By definition of $N_{S_{\tilde{\N}}}$, $e(n)\neq\emptyset$. Fix $w\in e(n)$; because $\tilde{Q}_a$ is serial, there is $w'$ such that $w\tilde{Q}_a w'$. It follows that $w'\tilde{Q}_a w'$ and $w'\in e(n)$.
So, $n\in\{([w']_i,i)\}_{i\in Ag}\in\F(S_{\tilde{\N},a})$, as desired.
	
	Of course, given any proper relational model 
	for introspective knowledge and belief,
	$\N$, we can use the above
	definitions to produce a simplicial belief model $\M_{\N}$.
	We now extend Theorem \ref{SimpTrans} in this context to the full language including belief modalities.
	
	
	
	\begin{lemma}
		Let $\N$ be a proper relational model for introspective knowledge and belief (over $Ag$). For all formulas $\phi \in \L_{KB}(Ag)$, for all $w \in W$,
		$$\N,w \models \phi \textrm{ iff } \M_{\N}, f(w) \models \phi.$$
	\end{lemma}
	\begin{proof}
		The proof is identical to that of Theorem \ref{SimpTrans} except for the inductive step
		for the
		belief modalities,
		which must be added.
		So suppose inductively, as in that proof, that for all $w \in W$,
		$\N,w \models \phi \textrm{ iff } \M_{\N}, f(w) \models \phi$.
		We wish to show that for all $w \in W$, $\N,w \models B_{a}\phi \textrm{ iff } \M_{\N}, f(w) \models B_{a}\phi$.
		
		Fix $w\in W$ and suppose
		$\N,w\vDash B_a\varphi$.
		Let $F\in\mathcal{F}(S_{\N,a})$ be such that $\pi_a(F)=\pi_a(f(w))$.
		It follows that $f^{-1}(F)Q_af^{-1}(F)$ and
		$w R_{a} f^{-1}(F)$.
		Because $Q_a$ is constant on $R_a$ equivalence classes,
		$Q_{a}(w) = Q_{a}(f^{-1}(F))$;
		thus,
		$w Q_{a} f^{-1}(F)$,
		and so $\N,f^{-1}(F)\models\varphi$. By the inductive hypothesis, we have that $\M_{\N},F\models\varphi$. It follows that $\M_{\N},f(w)\models B_a\varphi$, as desired.
		
		Suppose now that $\M_{\N}, f(w) \models B_a\phi$.
		Let $u \in W$ be
		such that $wQ_au$.
		By the Euclidean property
		we have that $uQ_au$. It follows that $f(u)\in\mathcal{F}(S_{\N,a})$. Moreover, $wQ_au$ implies that $wR_au$, and so $\pi_a(f(u))=\pi_a(f(w))$. It follows that $\M_{\N},f(u)\models \varphi$. By the inductive hypothesis, we have that $\N,u\models\varphi$. Therefore, $\N,w\models B_a\varphi$, as desired.
	\end{proof}

	We can now
	establish
	completeness. Suppose for simplicial belief models we have $\Gamma\models\varphi$, and suppose for contradiction that $\Gamma\nvdash_{\full}\varphi$. Let $w$ be
	a maximal $\full$-consistent
	set
	extending $\Gamma \cup \{\lnot \phi\}$ (obtained, as usual, via Lindenbaum's Lemma).
	Then $\N,w\models\Gamma\cup\{\neg\varphi\}$. It
	then follows from the above
	that (for any $u$), $\M_{\tilde{\N}},f((w,u))\models\Gamma\cup\{\neg\varphi\}$, a contradiction.
\end{proof}

\section{Related and Future Work} \label{sec:Conc}


The simplicial representation of belief we introduce in
this paper stands in contrast to
two existing approaches.
In ``Knowledge and Simplicial Complexes'' \cite{KaSC}, standard UCF simplicial models $(N,V,S,L)$ are augmented with idempotent functions $f_{a} \colon V^{-1}(a) \to V^{-1}(a)$, one for each agent $a \in Ag$. These functions are used to define belief via the following clause:
\begin{align*}
	\M,X & \models B_a \phi \text{ iff } \forall Y \in \F(S) \text{ if } \pi_a(Y) = f_{a}(\pi_a(X)) \text{ then } \M,Y \models \phi.
\end{align*}
Loosely speaking, whereas a node $n \in V^{-1}(a)$ represents a ``perspective'' of agent $a$ for the purposes of evaluating knowledge modalities at facets containing $n$, for the purposes of evaluating \textit{belief} modalities at facets containing $n$, this semantics instead uses the node $f_{a}(n)$. In other words, agent $a$'s beliefs at $X$ are given by her knowledge at $f_{a}(\pi_{a}(X))$.
Because each $f_a$ is idempotent, it is easy to check that the $B_a$ modalities satisfy the \textsf{KD45} axioms.

However, this notion of belief
has essentially no relationship to knowledge; in particular,  
$K_a\varphi\rightarrow B_a\varphi$ is not valid.
A
natural solution to this might be to strengthen the semantics
so as to validate this formula. Unfortunately, this results in triviality, for it requires the following to hold for all $X\in\mathcal{F}(S)$:

$$\{Y \in \mathcal{F}(S) \: : \: f_a(\pi_a(X)) = \pi_a(Y)\} \subseteq \{Y \in \mathcal{F}(S) \: : \:\pi_a(X) = \pi_a(Y)\};$$
under the UCF assumption, this
containment does not hold in general unless
$f_a$ is the identity function.
In this case, of course,
belief and knowledge coincide.
Thus, this model for belief cannot satisfy both
UCF and the validity of ``knowledge implies belief'' without trivializing belief.

We find a quite different approach to defining belief in simplicial complexes
in ``Simplicial Belief''. \cite{SimpBel} The idea here depends on breaking the UCF assumption
so that facets may contain one \textit{or more} $a$-colored nodes, and then using the number of nodes a given facet contains to determine whether or not it is ``accessible'' for the purposes of evaluating belief. Specifically,
given $X \in \mathcal{F}(S)$, let $m_a(X):=|\{n\in X \: : \: V(n)=a\}|$
count the
number
of $a$-colored nodes in $X$,
and for
any $X,Y\in\mathcal{F}(S)$,
define
$$X \leq_a Y \text{ iff } m_a(X) \leq m_a(Y).$$
Write
$X \sim_a Y$ if and only if $\exists n\in X\cap Y$ such that $V(n)=a$,
and define $\trianglelefteq_a \; = \; \leq_a \cap \sim_a$.
These relations are used to define belief via the following clause:
\begin{align*}
	\M,X \models B_a \phi \text{ iff } \forall Y \in \F(S) \text{ if } Y \trianglelefteq_{a} X \text{ then }\M,Y \models \phi.
\end{align*}

This says that
if
$\varphi$ is true at all facets
$Y$ whose intersection with $X$ contains an $a$-colored node, and such that the number of $a$-colored nodes in $Y$ is bounded above by the number in $X$,
then $B_a\varphi$ is true at $X$.\footnote{The paper also provides a stronger notion of belief, taking the belief relation to access worlds with a \textit{minimum} number of perspectives, rather than all those with fewer than the starting world. The same result we show for the given definition also applies to this one.}
\commentout{
	They also take a second, stronger notion of belief:
	
	\begin{align*}
		\mathcal{M},X\models B_a\varphi\text{ iff }\forall Y\in\mathcal{F}(S) & \text{ if }X\trianglerighteq_a Y\\
		& \text{ and }\forall Z\in\mathcal{F}(S)(\text{ if }X\trianglerighteq_a Z\text { then }m_a(Z)\geq m_a(Y))\\
		& \text{ then }\mathcal{M},Y\models\varphi
	\end{align*}
	
	This second definition says that if $\varphi$ is true at all facets sharing an $a$-intersection with $X$ with a minimal number of perspectives over this collection, then $B_a\varphi$ is true at $X$.
}
The rough intuition here, as articulated in \cite{SimpBel}, is that facets containing more $a$-perspectives are less plausible (to agent $a$) than those containing fewer, since $a$ is in some sense less certain of their own state of mind in the former.

It's easy to see that this definition validates
$K_{a} \phi \lthen B_{a}\phi$.
And it is certainly interesting from a formal perspective to relax the UCF assumption and leverage this extra freedom to define belief. However, the epistemic interpretation of facets with multiple $a$-perspectives remains somewhat murky.
Unfortunately, if we re-impose the UCF condition, this of course implies that for all $X \in \F(S)$, $m_a(X)=1$, making $\leq_a$ the total relation on $\mathcal{F}(S)$, hence $\trianglelefteq_a \; = \; \sim_a$.
In this case again we see that belief and knowledge coincide.

In contrast to these approaches, our framework provides a definition of belief in simplicial complexes satisfying \textit{all} three of: (1) UCF, (2) validity of $K_{a} \phi \lthen B_{a} \phi$ is valid, and (3) $B_{a} \phi \lthen K_{a} \phi$ is \textit{not} valid.
This has a few further upshots. The usual techniques for translating frame models into simplicial models apply very cleanly (indeed, almost trivially) in our setting,
yielding soundness and completeness with respect to the simple $\full$ axiom system (Theorem \ref{Thm:Comp}).
The only wrinkle in the proof is the condition that the underlying
relational
frame be proper,
which is essentially circumvented via Theorem \ref{thm:note}.

This approach is similar to the ``directed'' approach taken in David Lehnherr's PhD thesis ``Simplicial Structures for Epistemic Reasoning in
Multi-agent Systems''. \cite{Lehnherr_2025} In this thesis, a knowledge model $\M$ is augmented with a function $\rho \colon V \times S \rightarrow 2$, which specifies whether or not a face in $S$ is considered ``possible'' from a node in $V$.
It is assumed that if $\rho(n,F)=1$, then there is $F'\in\F(S)$ such that $n\in F'$, and if $G\subseteq F$, $\rho(n,G)=1$.
The semantic clause for belief is given by
$\M,X\vDash B_a\varphi$ if and only if for all $Y$ such that $\rho(\pi_a(X),Y)=1$, $\M,Y\vDash\varphi$.

Our approach makes a simplifying assumption to emphasize the key notion of belief subcomplexes, so we can explore their contribution more directly.
Namely,
we do not assume that variables are assigned to nodes, but rather to facets.\footnote{Future work in our belief semantics will re-introduce this assumption.} Nevertheless, we can define the collection of subcomplexes
and $\rho$ interchangeably, at least for facets. For the first direction, $S_a:=\{F\in\F(S)~|~\rho(\pi_a(F),F)=1\}$. For the other direction, if $F\in\F(S)$, say that $\rho(\pi_a(F),F)=1$ if and only if $F\in\F(S_a)$, and $0$ otherwise.

There is another, more recent approach which attempts to model belief using chromatic hypergraphs. \cite{HGD} Given $N$ and $V$ as above, a
\defin{chromatic hypergraph} $C$ is a subset of $2^N\times 2^N$, called ``hyper edges'', which are the stand-in for worlds.
For each $(X,Y)\in C$, $X\cap Y=\emptyset$. Moreover, if $n,n'\in X\cup Y$, and if $V(n)=V(n')$, then $n=n'$. Hence, each ordered pair $(X,Y)\in C$ is uniquely colored across the union of $X$ and $Y$. It is also assumed that, for each $n\in N$, there is a pair $(X,Y)$ such that $n\in X\cup Y$, and $|X\cup Y|=|Ag|$. We can think of those pairs $(X,Y)$ where $|X\cup Y|=|Ag|$ as the equivalent of ``facets'' for chromatic hypergraphs. And, the disjoint pair $X$ and $Y$ divide each ``facet'' into two pieces. This is how doxastic information is read off of a chromatic hypergraph. $X$ is the set of nodes that consider the whole hyper edge $(X,Y)$ ``doxastically accessible'' or possible. Consider two hyper edges $(X,Y)$ and $(X',Y')$. If $n\in(X\cup Y)\cap X'$, then $V(n)$ considers $(X',Y')$ possible from $(X,Y)$. While this approach is promising, it requires
structural assumptions significantly more elaborate than the usual presentation of simplicial semantics.
We hope then that our approach
provides a simpler and more streamlined way of incorporating belief into simplicial semantics.

As mentioned above, the existing literature has largely taken properness as a technical condition, introduced in order to facilitate translations from relational to simplicial models.
\cite{KaSC,DoA,Death,SimpDEL}
However, as noted at the end of Section \ref{sec:prop}, in order to ``make'' a model proper we have to multiply worlds in a way that produces substantial redundancy. As a result, the structures become harder to parse---both the relational models and their simplicial correlates. We suspect, for example, that few will find the simplicial model in Figure \ref{Figure3} as easy to understand as the relational model in Figure \ref{fgr:3ag}.

It is natural to wonder, then, if there is any way to avoid imposing properness altogether. One promising approach utilizes the notion of \emph{simplicial sets}.\footnote{The
	idea of using simplicial sets
	in the context of simplicial models
	was introduced in \cite{SimpSet}. Recent work expanding this idea is contained in \cite{CACHIN2025114902}}
Roughly speaking, simplicial sets generalize simplicial complexes by allowing us to count faces \textit{with multiplicity},
similar to the generalization of a set to a multiset, or a graph to a multigraph. For instance,
in a simplicial
complex,
the triangle $\{a,b,c\}$ is uniquely determined by
the three nodes it contains, whereas
in a simplicial set, there could be multiple triangles
on the same underlying set of nodes
$\{a,b,c\}$.
Visually, one may picture multiple triangular ``membranes'' connecting these nodes.
Informally, we may
define a simplicial set as a ``multiset of multisets closed under subsets''.

Using this structure,
one can capture the model
in Figure \ref{fgr:3ag} with a
simplicial set $S$ together with subsimplicial sets $S_{a}$, $S_{b}$, and $S_{c}$, as depicted
in Figure \ref{fgr:3agsimpset}.

\begin{figure}[htbp]
\begin{center}
	\begin{tikzcd}[scale cd=1.5]
		& \color{red}a_1 \arrow[ldd, no head] \arrow[violet, ldd, no head, bend right] \arrow[yellow, ldd, no head, bend right=49] &                                                                                                             \\
		&                                                                                                          &                                                                                                             \\
		\color{blue}b_1 \arrow[rr, no head] \arrow[violet, rr, no head, bend right] \arrow[yellow, rr, no head, bend right=49] &                                                                                                          & \color{green}c_1 \arrow[luu, no head] \arrow[violet, luu, no head, bend right] \arrow[yellow, luu, no head, bend right=49]
	\end{tikzcd}
	\caption{A translation of the model from Figure \ref{fgr:3ag} into a simplicial
		simplicial set $S$ together with subsimplicial sets $S_{a}$, $S_{b}$, and $S_{c}$
		(violet facets belong to $S_a$ and $S_b$, yellow facets belong to $S_a$ and $S_c$, while black facets belong to only $S$)}
	\label{fgr:3agsimpset}
\end{center}
\end{figure}
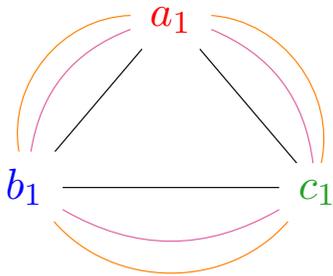\label{Figure4}

Intuitively,
the node $a_1$ corresponds to the
(unique!)
perspective that agent $a$ has in Figure \ref{fgr:3ag}, namely the worlds $w_1$ and $w_2$. Similarly, $b_1$ corresponds to
the view that only $w_{1}$ is possible,
and $c_1$ corresponds to only seeing $w_2$.
But we now have the freedom to associate not one but three facets with this set of perspectives.
The black facet corresponds with $w_0$, the world accessible to none of the three agents. The violet facet
belongs to
$S_a$ and $S_b$,
but not $S_{c}$, corresponding to
$w_1$, the world accessible to agents $a$ and $b$
but not $c$. And similarly the yellow facet
belongs to
$S_a$ and $S_c$,
but not $S_{b}$, corresponding to
$w_2$.
Thus we see that simplicial sets give us the resources to represent non-proper models without artificially proliferating perspectives; this provides a promising direction for future research in simplicial semantics and particularly in the representation of belief therein.

\bibliographystyle{eptcs.bst}	
\bibliography{Epistemology}

\end{document}